\documentclass[11pt,letterpaper]{article}

\usepackage[T1]{fontenc}
\usepackage{lmodern}
\usepackage{microtype}
\usepackage[margin=1in]{geometry}
\usepackage{amsmath,amssymb,amsthm,mathtools}
\usepackage{booktabs}
\usepackage{enumitem}
\usepackage{algorithm}
\usepackage[noend]{algpseudocode}
\usepackage{xcolor}
\usepackage{url}
\usepackage{aliascnt}
\usepackage[colorlinks=true,citecolor=violet]{hyperref}
\usepackage[nameinlink,capitalise,noabbrev]{cleveref}


\newtheorem{theorem}{Theorem}[section]

\newaliascnt{lemma}{theorem}
\newtheorem{lemma}[lemma]{Lemma}
\aliascntresetthe{lemma}

\newaliascnt{proposition}{theorem}

\aliascntresetthe{proposition}

\newaliascnt{corollary}{theorem}
\newtheorem{corollary}[corollary]{Corollary}
\aliascntresetthe{corollary}

\newaliascnt{claim}{theorem}

\aliascntresetthe{claim}

\theoremstyle{definition}
\newaliascnt{definition}{theorem}

\aliascntresetthe{definition}

\newaliascnt{remark}{theorem}
\newtheorem{remark}[remark]{Remark}
\aliascntresetthe{remark}

\crefname{theorem}{Theorem}{Theorems}
\Crefname{theorem}{Theorem}{Theorems}
\crefname{lemma}{Lemma}{Lemmas}
\Crefname{lemma}{Lemma}{Lemmas}
\crefname{proposition}{Proposition}{Propositions}
\Crefname{proposition}{Proposition}{Propositions}
\crefname{corollary}{Corollary}{Corollaries}
\Crefname{corollary}{Corollary}{Corollaries}
\crefname{claim}{Claim}{Claims}
\Crefname{claim}{Claim}{Claims}
\crefname{remark}{Remark}{Remarks}
\Crefname{remark}{Remark}{Remarks}

\newcommand{\cut}{\mathsf{cut}}

\newcommand{\Reach}{\operatorname{Reach}}

\newcommand{\eps}{\varepsilon}

\newcommand{\cT}{\mathcal{T}}

\newcommand{\defeq}{\vcentcolon=}
\newcommand{\intervalP}[2]{P_{(#1,#2]}}
\newcommand{\Cycle}{\textnormal{\textsc{cycle}}}

\title{Reachability in Directed Acyclic Graphs\\with Near-Linear Cut Queries}
\author{Sanjeev Khanna\thanks{Courant Institute, Warren Weaver Hall, New York University, New York, NY 10012.  Supported in part by NSF award CCF-2625203 and AFOSR award FA9550-25-1-0107. Email: {\tt sanjeev.khanna@nyu.edu}.} \and Aaron Putterman\thanks{School of Engineering and Applied Sciences, Harvard University, Cambridge, Massachusetts, USA. Supported in part by the Simons Investigator Awards of Madhu Sudan and Salil Vadhan, AFOSR award FA9550-25-1-0112, and a Jane Street Graduate Research Fellowship. Email: \texttt{aputterman@g.harvard.edu}.} \and Junkai Song\thanks{Courant Institute, Warren Weaver Hall, New York University, New York, NY 10012. Supported in part by NSF award CCF-2625203. Email: \texttt{junkaisong@nyu.edu}.}}
\date{July 2026}

\begin{document}
\maketitle

\begin{abstract}
    In the cut-query model, an algorithm is given access to a graph $G = (V, E)$ \emph{only} via cut queries. This model has seen significant attention in the undirected graph setting, with works establishing $O(n)$ cut query algorithms for computing the global minimum cut, $\widetilde{O}(n^{3/2})$ cut query algorithms for all pairs minimum cut, and many more. However, despite this vast array of progress in designing sub-quadratic query algorithms for computing properties of undirected graphs, there has been \emph{no} progress in designing such algorithms in directed graphs. Indeed, even for basic problems like whether a vertex $t$ is reachable from a vertex $s$, the cut query complexity is only known to be bounded in the interval $[\Omega(n), O(n^2 / \log n)]$.

    In this work, we begin a systematic study of these basic problems in directed \emph{acyclic} graphs (DAGs). In this setting, we show that reachability from a single vertex and even topological sorting are both computable in $O(n \log^3 n)$ many cut queries.
    As a consequence, we also obtain an algorithm which, for any \emph{arbitrary} directed graph $G$, uses only $O(n \log^3 n)$ cut queries and determines whether $G$ contains a cycle.
\end{abstract}

\section{Introduction}

A basic goal in theoretical computer science is to design efficient algorithms for computing properties of graphs. Naturally, such algorithms themselves can vary significantly depending on \emph{how} the graph is represented. 
A recent line of study, first proposed in the work of Rubinstein, Schramm, and Weinberg \cite{RSW18}, studies the design of graph algorithms when the only access to the graph $G = (V, E)$ is in terms of \emph{cut queries}. In this model, the algorithm can query the \emph{cut-oracle} with any subset of vertices $S \subseteq V$, and the oracle reports the \emph{number} of edges that are cut by this set $S$. 

To demonstrate the power of these cut queries, one can observe that in undirected graphs, using only $\binom{n}{2} + n$ queries suffices for \emph{exactly} reconstructing the entire graph at hand. Indeed, one simply queries all singleton vertices $v \in V$ and all pairs of vertices $\{u,v \} \in \binom{V}{2}$. Now, to test for the presence of any single edge $(u,v) \in E$, one must only check whether 
\[
\cut(\{u\}) + \cut(\{v\}) = \cut(\{u,v \}).
\]

With this simple reconstruction baseline established, the key question in works studying the cut-query model is to determine which algorithmic problems can be solved using only \emph{sub-quadratically} many cut queries. Indeed, already in this aforementioned work of Rubinstein, Schramm, and Weinberg \cite{RSW18}, they showed that the global minimum cut value of an undirected simple graph $G$ can be computed using $\widetilde{O}(n)$ randomized cut queries, and any $s$--$t$ minimum cut value can be computed using $\widetilde{O}(n^{5/3})$ randomized cut queries.

These bounds have been improved in subsequent works. \cite{MN20} gave a randomized algorithm using $\widetilde{O}(n)$ cut queries for the global minimum cut in undirected \emph{weighted} graphs, while \cite{AEGLMN22} improved the query complexity to $O(n)$ for undirected simple graphs. \cite{ASW25} obtained the first \emph{deterministic} algorithm for global minimum cut and $s$--$t$ minimum cut using $\widetilde O(n^{5/3})$ queries. \cite{JNS26} improved the randomized query complexity for exact $s$--$t$ minimum cut to $\widetilde{O}(n^{8/5})$, and \cite{KK26} gave a randomized algorithm for computing all-pairs minimum cuts using only $\widetilde{O}(n^{3/2})$ queries.
This line of work has also expanded to efficient query algorithms for other algorithmic tasks. 
For instance, \cite{PRW24} presented an algorithm to compute $(1/2 - \eps)$-approximate maximum cut in $O(\log n)$ cut queries, and $(1-\eps)$-approximate maximum cut in $\widetilde{O}(n)$ cut queries.

Despite this vast array of work presenting efficient (sub-quadratic) cut-query algorithms, there is a notable, glaring omission. Namely, in the setting of \emph{directed graphs}, even the most basic connectivity questions are not well-understood. In this setting, each edge is \emph{directed} from a vertex $u$ (tail) to a vertex $v$ (head), and an edge $(u \rightarrow v)$ is \emph{cut} by a set $S$ if $u \in S$ and $v \in \bar{S}$. It is not hard to see that in this model, cut queries do not suffice for explicitly recovering the underlying graph; indeed, one can consider a cycle of length $3$ on vertices $a, b, c$. Both orientations of the cycle $a \rightarrow b \rightarrow c \rightarrow a$ and $a \rightarrow c \rightarrow b \rightarrow a$ yield the same cut query values for all cuts $S \subseteq V$. Nevertheless, if two directed graphs have the exact same cut-values for every cut $S \subseteq V$, then they also have the \emph{same} reachability structure. A vertex $b$ is reachable from a vertex $a$ if and only if there is no cut $S$ such that $a \in S$, $b \in \bar{S}$ and $\cut(S) = 0$. Beyond this, one can even observe that it is possible to reconstruct the \emph{entire} cut function of a directed graph in only $O(n^2)$ many directed cut queries (even $O(n^2 / \log n)$ many queries).\footnote{Note that this is not a formally published result, see \url{https://sites.google.com/site/dannanongkai/open?authuser=0} and \url{https://nextcloud.mpi-klsb.mpg.de/index.php/s/jBobpANjnjeJo3w?dir=/&editing=false&openfile=true} for discussion.}

Despite this, no nontrivial directed graph problem is currently known to admit a sub-quadratic cut-query algorithm. Even the seemingly simplest problem, $s$--$t$ reachability, remains open:
\begin{quote}
    \emph{Given two vertices $s, t$, how many cut queries are required to determine if there is a path from $s$ to $t$?}
\end{quote}

This question is further motivated by its close connection to submodular function minimization (SFM). Given value oracle access to a submodular function $f:2^{[n]} \to \mathbb{R}$, the goal is to find a minimizer of $f$ using as few value queries as possible. The best known algorithm uses $\widetilde{O}(n^2)$ queries \cite{jiang2022minimizing}, while the best lower bound remains $\widetilde{\Omega}(n)$ \cite{CGJS22}. Since the cut function of a directed graph is submodular, understanding the query complexity of directed cut functions may provide insight to the complexity of SFM. 

In this work we resolve this reachability problem, up to polylogarithmic factors, on directed \emph{acyclic} graphs (DAGs).

\paragraph{Our Results}

 Our first theorem constructs a topological order of any DAG via only $\widetilde{O}(n)$ many cut queries:

\begin{theorem}[Topological ordering and cycle detection]\label{thm:topological-main}
 Let $G=(V,E)$ be a directed graph on $n$ vertices, given through its directed cut oracle.
	There is a deterministic algorithm using $O(n\log^3 n)$ cut queries that either
 \begin{enumerate}[label=(\roman*),nosep]
		\item outputs a topological ordering of $G$, or
		 \item correctly reports that $G$ contains a directed cycle.
	\end{enumerate}
 In particular, on a DAG the algorithm always outputs a topological ordering.
\end{theorem}

 Once an ordering is known, directed edge counts from one vertex into an arbitrary set can be
isolated with constantly many cut queries.  A binary search then finds one outgoing edge in
logarithmically many queries, which permits a standard depth-first search. Letting $\Reach_G(s)$ denote \emph{all} vertices reachable from a vertex $s$ in a graph $G$, we then have that:

\begin{theorem}[Reachability with a topological order]\label{thm:reach-main}
	Let $G=(V,E)$ be an $n$-vertex DAG, let $s\in V$, and suppose a topological ordering of $G$ is
	given.  There is a deterministic algorithm that returns $\Reach_G(s)$ using $O(n\log n)$ directed
	cut queries.
	\end{theorem}

Combining the two algorithms gives the following immediate consequence.
 \begin{corollary}\label{cor:reach-from-scratch}
Single-source reachability in an $n$-vertex DAG can be solved deterministically with
$O(n\log^3 n)$ directed cut queries.
\end{corollary}

\paragraph{Technical overview.}
The starting point of our cut query algorithm is Kahn's classical topological-sorting procedure, which repeatedly removes a vertex of residual in-degree
zero.  Naturally, the resulting sequence of peeled vertices constitutes a topological ordering.
Implementing such a procedure is not hard in the cut-query model. Indeed, each query of the form $\cut(V \setminus \{v\})$ exactly reports the in-degree of a vertex $v$, and thus quickly enables discovery of any in-degree $0$ vertex. For any such vertex, one can then identify all edges incident to $v$ via the \emph{undirected} edge recovery procedure. The key point is that because $v$ has in-degree $0$, given any neighbor $u$ of $v$, one can immediately infer the direction of the edge between them (namely, that the edge goes from $v$ to $u$).
The obstacle in the cut-query model is that explicitly updating all residual in-degrees after
each removal can cost quadratic queries, as this procedure is implicitly learning the entire ``undirectification'' of the graph $G$.
Rather than explicitly updating and maintaining all residual in-degrees, our algorithm relies on a lazy maintenance data structure.

As we peel vertices with residual in-degree $0$, we periodically update our estimates for the in-degrees of the remaining vertices.
A residual vertex $v$ whose last verified positive in-degree lies in $[2^i,2^{i+1})$ is assigned to
a bucket $A_i$.  Observe that for any such vertex, in order for it to become an in-degree $0$ vertex, it must then have $\geq 2^i$ of its ancestors peeled in our topological sort. It is exactly this observation which allows us to avoid \emph{explicitly} updating the residual in-degree of $v$ after every new vertex is peeled. 

In more detail, each bucket  $A_i$ owns a balanced binary tree whose leaves are the vertices, and whose interior nodes represent the union of all of its descendants.  Whenever a new
source (residual in-degree $0$ vertex) is peeled, the root of every tree (i.e., the set of all leaf nodes) is queried to determine the \emph{total} number of in-edges
that vertices in the bucket have lost since by removing this source.  A child of the root is subsequently opened only if at
least
\[
  \lambda_i = \max\left\{1,
    \left\lfloor \frac{2^i}{4(H+1)}\right\rfloor\right\},
  \qquad H=\left\lceil\log_2\max\{n,2\}\right\rceil,
\]
new deletions have accumulated in that child's subtree.  At a leaf (which represents a single vertex), the corresponding vertex's
current in-degree is then computed exactly; the vertex is either declared a source or moved to the
appropriate lower bucket. In this way, queries to a single vertex are delayed until we are sure that \emph{a substantial number} of in-neighbors of the vertex have been removed. At the same time, because of our tracking of the in-degree decrements at higher nodes in the tree, we can still be sure that until we query a leaf node, it must have in-degree $> 0$. 

The remaining subtlety in this plan is that a vertex can enter a bucket after some tree nodes were last opened.  Merely
counting all deletions since a node's last opening would then charge the vertex for deletions that
occurred before it joined the bucket.  We therefore store, for every active vertex and every
ancestor node, an offset equal to precisely this pre-entry contribution.  The resulting decrease
counter is exact.

Two complementary estimates yield the proof.  First, along a root-to-leaf path, every unopened
segment hides fewer than $\lambda_i$ deletions.  There are only $H+1$ such segments, so a vertex
remaining in $A_i$ still has positive in-degree.  Hence the algorithm never misses a residual
source.  Second, opening a node can be charged to $\lambda_i$ new in-edge deletions in its subtree.
A vertex contributes fewer than $2^{i+1}$ deletions while it belongs to $A_i$, which implies that a
node with $r$ leaves is opened only $O(rH)$ times.  Summing over levels and buckets gives
$O(nH^3)=O(n\log^3 n)$ queries.

\paragraph{Note.} The authors recently learned that this same problem has been independently and concurrently studied by two separate groups in addition to our own. Both groups (a) Ben Bals, Yasamin Nazari, and Matei Tinca and (b) Deeparnab Chakrabarty and Andrew Zhao discovered $O(n^{3/2} \sqrt{\log(n)})$ cut query algorithms for computing single-source reachability in DAGs via topological sorting. 

\section{Preliminaries}\label{sec:prelim}

All logarithms are base two.  The input is a finite loopless simple directed graph
$G=(V,E)$ with $|V|=n$.  For sets $X,Y\subseteq V$, write
\[
  E(X,Y) \defeq \{(x,y)\in E:x\in X,\ y\in Y\}.
\]
For a vertex $v$, we abbreviate $E(\{v\},Y)$ to $E(v,Y)$ and similarly for $E(X,v)$.
The directed cut oracle returns
\[
  \cut(S) \defeq |E(S,V\setminus S)|
\]
for each queried set $S\subseteq V$.  The query complexity counts oracle calls; all local
computation and storage are free.  The assumption that the graph is loopless is necessary for
cycle detection, since self-loops are invisible to every cut query.  Simplicity ensures that all
in-degrees are at most $n-1$; polynomially bounded arc multiplicities can be handled by adding
corresponding logarithmic scales.

We begin with the elementary identity that underlies every directional edge-count query in the
paper.

\begin{lemma}[Symmetric crossing identity]\label{lem:symmetric-crossing}
For disjoint sets $X,Y\subseteq V$,
\[
  \cut(X)+\cut(Y)-\cut(X\cup Y)
  = |E(X,Y)|+|E(Y,X)|.
\]
\end{lemma}

\begin{proof}
Let $R=V\setminus(X\cup Y)$.  Expanding the three cut values gives
\begin{align*}
\cut(X) &= |E(X,Y)|+|E(X,R)|,\\
\cut(Y) &= |E(Y,X)|+|E(Y,R)|,\\
\cut(X\cup Y) &= |E(X,R)|+|E(Y,R)|.
\end{align*}
Subtracting the last equality from the first two proves the identity.
\end{proof}

The topological-sorting algorithm peels vertices one at a time.  We use the following notation.
If $p_1,p_2,\ldots$ is the peel sequence, define
\[
  P_t \defeq \{p_1,\ldots,p_t\},
  \qquad P_0\defeq\varnothing,
  \qquad \intervalP{a}{b}\defeq P_b\setminus P_a.
\]
A vertex $p_t$ is a \emph{valid peel} if it has in-degree zero in $G[V\setminus P_{t-1}]$.

\begin{lemma}[One-way boundary of a source prefix]\label{lem:source-prefix}
Suppose $p_1,\ldots,p_t$ are valid peels.  Then
\[
  E(V\setminus P_t,P_t)=\varnothing.
\]
Consequently, for every $X\subseteq P_t$ and $Y\subseteq V\setminus P_t$,
\[
  |E(X,Y)| = \cut(X)+\cut(Y)-\cut(X\cup Y),
\]
and this quantity can be computed with at most three cut queries.
\end{lemma}

\begin{proof}
Fix an arc $(y,p_j)$ with $p_j\in P_t$.  If $y\notin P_t$, then $y\notin P_{j-1}$, so this arc would
enter $p_j$ from the residual graph at the moment $p_j$ was peeled, contradicting validity.
Thus no arc enters $P_t$ from its complement.  In particular, $E(Y,X)=\varnothing$ for
$X\subseteq P_t$ and $Y\subseteq V\setminus P_t$.  The formula follows from
\cref{lem:symmetric-crossing}.
\end{proof}

We refer to the three-query computation in \cref{lem:source-prefix} as
$\mathsf{DirCount}(X,Y)$.  Notice also that the initial in-degree of a vertex is available in one
query:
\begin{equation}\label{eq:initial-indegree}
  d^-_0(v) \defeq |E(V\setminus\{v\},v)| = \cut(V\setminus\{v\}).
\end{equation}

\section{Topological Sorting by Lazy Source Peeling}\label{sec:toposort}

This section proves \cref{thm:topological-main}.  The algorithm follows source peeling, but it
maintains residual in-degrees only approximately until a vertex becomes a plausible source.

\subsection{Buckets and refinement trees}\label{subsec:trees}

Set
\[
  H \defeq \left\lceil\log_2\max\{n,2\}\right\rceil
  \quad\text{and}\quad
  I\defeq\{0,1,\ldots,H-1\}.
\]
Since the graph is simple, every positive in-degree belongs to $[2^i,2^{i+1})$ for a unique
$i\in I$.

At every time $t$, the unpeeled vertices are partitioned as
\begin{equation}\label{eq:partition}
  V\setminus P_t = Z(t)\ \dot\cup\ A_0(t)\ \dot\cup\cdots\dot\cup\ A_{H-1}(t).
\end{equation}
The set $Z(t)$ contains vertices known to have residual in-degree zero.  A vertex in $A_i(t)$ was
most recently verified to have in-degree at least $2^i$ and less than $2^{i+1}$; its current degree
may have fallen below $2^i$, but the data structure has not yet necessarily inspected its leaf.
Because in-degrees only decrease, vertices move only from higher-indexed buckets to lower-indexed
buckets, and eventually to $Z$.

For each $i\in I$, let $\cT_i$ be a complete binary tree with
$N=2^H$ leaves.  The first $n$ leaves are labeled by the vertices and the remaining leaves are
dummies.  For a node $x\in\cT_i$, let $L(x)\subseteq V$ be the set of real vertex labels below
$x$.  Each tree has height $H$.

The threshold for bucket $i$ is
\begin{equation}\label{eq:lambda}
  \lambda_i \defeq
  \max\left\{1,\left\lfloor\frac{2^i}{4(H+1)}\right\rfloor\right\}.
\end{equation}

Each node $x\in\cT_i$ stores a timestamp $\mathsf{hist}_i(x)$, the most recent iteration in which
$x$ was \emph{synchronized}.  For every active vertex $v\in A_i(t)\cap L(x)$, it also stores an
offset $\mathsf{off}_i(x,v)$.  Let $\mathsf{ent}_i(v)$ denote the iteration in which $v$ entered its
current tenure in $A_i$.

Initially, all histories are zero.  Every vertex of initial in-degree zero is placed in $Z(0)$.
Every other vertex $v$ is placed in the unique bucket
$i=\lfloor\log d^-_0(v)\rfloor$, with $\mathsf{ent}_i(v)=0$ and all its offsets set to zero.

\paragraph{Entering a bucket.}
Suppose $v$ enters $A_i$ at time $t$.  For every node $x$ on the root-to-$v$ path in $\cT_i$, set
\begin{equation}\label{eq:offset-entry}
  \mathsf{off}_i(x,v)
  \defeq
  \bigl|E(\intervalP{\mathsf{hist}_i(x)}{t},v)\bigr|.
\end{equation}
By \cref{lem:source-prefix}, each value in \eqref{eq:offset-entry} costs $O(1)$ cut queries.

\paragraph{Synchronizing a node.}
To synchronize $x\in\cT_i$ at time $t$, set
\begin{equation}\label{eq:synchronize}
  \mathsf{hist}_i(x)\gets t,
  \qquad
  \mathsf{off}_i(x,v)\gets 0
  \quad\text{for all }v\in A_i(t)\cap L(x).
\end{equation}
This operation uses no oracle queries.

\paragraph{The decrease counter.}
For a node $x\in\cT_i$, define
\begin{align}
  \mathsf{Decrease}_i(x,t)
  \defeq{}&
  \bigl|E(\intervalP{\mathsf{hist}_i(x)}{t},
      A_i(t)\cap L(x))\bigr| \notag\\
  &\quad -
  \sum_{v\in A_i(t)\cap L(x)}\mathsf{off}_i(x,v).
  \label{eq:decrease-def}
\end{align}
The first term is obtained by one call to $\mathsf{DirCount}$, hence with $O(1)$ cut queries.
The offsets are locally stored values.

\begin{lemma}[Exact lazy accounting]\label{lem:decrease-exact}
For every bucket $i$, node $x\in\cT_i$, and iteration $t$,
\begin{equation}\label{eq:decrease-exact}
\mathsf{Decrease}_i(x,t)
=
\sum_{v\in A_i(t)\cap L(x)}
  \bigl|E(\intervalP{\max\{\mathsf{hist}_i(x),\mathsf{ent}_i(v)\}}{t},v)\bigr|.
\end{equation}
In particular, it is a nonnegative integer equal to the number of in-edges lost by the currently
active vertices below $x$, counting only losses after both the node's last synchronization and the
vertex's entry into $A_i$.
\end{lemma}

\begin{proof}
Fix $v\in A_i(t)\cap L(x)$.  If $v$ was already active when $x$ was last synchronized, then
$\mathsf{ent}_i(v)\le \mathsf{hist}_i(x)$ and \eqref{eq:synchronize} set
$\mathsf{off}_i(x,v)=0$.  Its contribution to \eqref{eq:decrease-def} is therefore
$|E(\intervalP{\mathsf{hist}_i(x)}{t},v)|$.

If $v$ entered later, then \eqref{eq:offset-entry} set
\[
  \mathsf{off}_i(x,v)
  =|E(\intervalP{\mathsf{hist}_i(x)}{\mathsf{ent}_i(v)},v)|.
\]
Subtracting this value from
$|E(\intervalP{\mathsf{hist}_i(x)}{t},v)|$ leaves exactly
$|E(\intervalP{\mathsf{ent}_i(v)}{t},v)|$.  Summing over active vertices proves the lemma.
\end{proof}

\subsection{Opening nodes}\label{subsec:opening}

When a node $x\in\cT_i$ is opened at time $t$, it is first synchronized.  If $x$ is internal, its
two children are tested and recursively opened whenever their decrease counters reach
$\lambda_i$.  If $x$ is a leaf labeled by an active vertex $v$, the algorithm computes the exact
residual in-degree
\begin{equation}\label{eq:current-degree}
  d^-_t(v)
  = d^-_0(v)-|E(P_t,v)|.
\end{equation}
The second term is an $O(1)$-query $\mathsf{DirCount}$ computation.  If
$d^-_t(v)\ge 2^i$, the vertex remains in $A_i$.  Otherwise it is removed from $A_i$: it enters
$Z$ when $d^-_t(v)=0$, and enters bucket
$A_{\lfloor\log d^-_t(v)\rfloor}$ when $d^-_t(v)>0$.

\begin{algorithm}[H]
\caption{$\mathsf{Open}(i,x,t)$}\label{alg:open}
\begin{algorithmic}
  \State Synchronize $x\in\cT_i$ at time $t$ according to \eqref{eq:synchronize}.
  \If{$x$ is a dummy leaf}
    \State \Return
  \ElsIf{$x$ is a leaf labeled by $v$}
    \If{$v\notin A_i(t)$}
      \State \Return
    \EndIf
    \State Compute $d^-_t(v)=d^-_0(v)-|E(P_t,v)|$.
    \If{$d^-_t(v)<2^i$}
      \State Remove $v$ from $A_i(t)$.
      \If{$d^-_t(v)=0$}
        \State Insert $v$ into $Z(t)$.
      \Else
        \State $j\gets\lfloor\log d^-_t(v)\rfloor$.
        \State Insert $v$ into $A_j(t)$ and initialize its offsets by \eqref{eq:offset-entry}.
      \EndIf
    \EndIf
    \State \Return
  \EndIf
  \For{each child $y$ of $x$}
    \If{$\mathsf{Decrease}_i(y,t)\ge\lambda_i$}
      \State $\mathsf{Open}(i,y,t)$.
    \EndIf
  \EndFor
\end{algorithmic}
\end{algorithm}

Synchronizing leaves in \cref{alg:open} is essential: without this step, the same old deletions can
be counted repeatedly each time a leaf is reconsidered.

\subsection{The complete algorithm}\label{subsec:complete-alg}

The main procedure first queries all initial in-degrees and initializes the partition
\eqref{eq:partition}.  It then repeatedly removes an arbitrary vertex from $Z$.  After the removal,
the root of every bucket tree is tested; a root is opened when its counter reaches the bucket
threshold.

\begin{algorithm}[H]
\caption{$\mathsf{TopologicalOrderOrCycle}(G)$}\label{alg:topological}
\begin{algorithmic}
  \State Query $d^-_0(v)=\cut(V\setminus\{v\})$ for every $v\in V$.
  \State Initialize $Z$, the buckets $A_i$, the trees $\cT_i$, all histories, and all offsets.
  \State $P_0\gets\varnothing$ and $t\gets 0$.
  \While{$P_t\neq V$}
    \If{$Z(t)=\varnothing$}
      \State \Return \Cycle.
    \EndIf
    \State Choose an arbitrary $v\in Z(t)$ and remove it from $Z(t)$.
    \State Set $p_{t+1}\gets v$, $P_{t+1}\gets P_t\cup\{v\}$, and $t\gets t+1$.
    \For{$i=0,1,\ldots,H-1$}
      \If{$\mathsf{Decrease}_i(\operatorname{root}(\cT_i),t)\ge\lambda_i$}
        \State $\mathsf{Open}(i,\operatorname{root}(\cT_i),t)$.
      \EndIf
    \EndFor
  \EndWhile
  \State \Return $(p_1,\ldots,p_n)$.
\end{algorithmic}
\end{algorithm}

\subsection{Correctness}\label{subsec:correctness}

The use of \(\mathsf{DirCount}\) is justified inductively.  At the beginning of iteration \(t+1\),
we assume that \(p_1,\ldots,p_t\) are valid peels.  Then \cref{lem:source-prefix} makes every
set-to-set count used during that iteration exact.  The lemmas below show that the maintained
source set is consequently exact, so the next selected vertex is again a valid peel.  The induction
starts at \(t=0\).

The central invariant is that an active vertex cannot silently lose all its incoming edges.

\begin{lemma}[Positive-degree invariant]\label{lem:positive-invariant}
After all tree processing in iteration $t$, every active vertex $v\in A_i(t)$ satisfies
\begin{equation}\label{eq:degree-lower-bound}
  d^-_t(v)
  \ge
  2^i-(H+1)(\lambda_i-1)
  >0.
\end{equation}
\end{lemma}

\begin{proof}
Fix $v\in A_i(t)$.  Let
$x_0,x_1,\ldots,x_H$ be the root-to-leaf path for $v$ in $\cT_i$, with $x_0$ the root and $x_H$
the leaf.  Define
\[
  \tau_j
  \defeq
  \max\{\mathsf{hist}_i(x_j),\mathsf{ent}_i(v)\}
  \qquad (0\le j\le H).
\]
A child can be synchronized only during a recursive call made while its parent is opened.  Hence
$\mathsf{hist}_i(x_0)\ge\cdots\ge\mathsf{hist}_i(x_H)$, and therefore
\[
  \tau_0\ge\tau_1\ge\cdots\ge\tau_H.
\]
At time $\tau_H$, either $v$ has just entered $A_i$, or its leaf has just been opened and its exact
in-degree checked.  In both cases its residual in-degree at that time is at least $2^i$.

Consider an interval $(\tau_{j+1},\tau_j]$ with $\tau_{j+1}<\tau_j$.  Then $x_j$ was opened at
time $\tau_j$, but $x_{j+1}$ was not: otherwise the child's history would be at least $\tau_j$.
Thus
$\mathsf{Decrease}_i(x_{j+1},\tau_j)<\lambda_i$.  By
\cref{lem:decrease-exact}, the number of in-edges lost by $v$ during this interval is at most
$\lambda_i-1$.

The same argument applies to the final interval $(\tau_0,t]$.  The root is tested at the end of
every iteration.  If it was opened after $\tau_0$, then $\tau_0$ would be later; otherwise its
counter at time $t$ is below $\lambda_i$.  Hence $v$ loses at most $\lambda_i-1$ in-edges in this
final interval as well.

There are at most $H+1$ intervals.  Starting from degree at least $2^i$ at time $\tau_H$, we obtain
the first inequality in \eqref{eq:degree-lower-bound}.  If $\lambda_i=1$, the subtractive term is
zero.  Otherwise,
\[
  \lambda_i-1
  < \frac{2^i}{4(H+1)},
\]
so $(H+1)(\lambda_i-1)<2^i/4$.  This proves strict positivity.
\end{proof}

\begin{lemma}[The source set is exact]\label{lem:source-set}
After initialization and after every completed iteration $t$,
\[
  Z(t)=\{v\in V\setminus P_t:d^-_t(v)=0\}.
\]
\end{lemma}

\begin{proof}
A vertex enters $Z$ only after its exact residual in-degree is computed to be zero.  Residual
in-degrees never increase, so it remains a source until it is peeled.  Conversely, by the partition
\eqref{eq:partition}, every unpeeled vertex outside $Z$ lies in some $A_i$, and
\cref{lem:positive-invariant} gives it positive residual in-degree.
\end{proof}

\begin{lemma}[Correctness of source peeling]\label{lem:topological-correctness}
Algorithm \ref{alg:topological} returns a topological ordering when $G$ is acyclic and returns
\Cycle{} only when $G$ contains a directed cycle.
\end{lemma}

\begin{proof}
By \cref{lem:source-set}, every selected vertex is a source of the current residual graph.  Thus
all peels are valid.  If the algorithm peels every vertex, the resulting order is topological: an
arc from a later vertex into an earlier vertex would have entered the earlier vertex at the moment
it was peeled.

Suppose instead that the algorithm returns \Cycle.  The residual graph is nonempty and,
by \cref{lem:source-set}, has no vertex of in-degree zero.  Starting from any residual vertex and
repeatedly following an incoming arc must eventually revisit a vertex, producing a directed cycle.
Conversely, every nonempty DAG has a source, so on a DAG the algorithm cannot return
\Cycle{} before all vertices are peeled.
\end{proof}

\subsection{Query complexity}\label{subsec:query-complexity}

It remains to amortize the tree activity.  For a node $x$, write
$r(x)=|L(x)|$ for its number of real leaves.

\begin{lemma}[Opening bound]\label{lem:opening-bound}
For a fixed bucket $i$ and node $x\in\cT_i$, the number of times $x$ is opened is
$O((H+1)r(x))$.
\end{lemma}

\begin{proof}
If $r(x)=0$, then every decrease counter below $x$ is identically zero and $x$ is never opened.
Assume henceforth that $r(x)\ge 1$.

A vertex enters $A_i$ at most once, because residual in-degrees only decrease.  At entry its degree
is less than $2^{i+1}$.  Therefore, while active in $A_i$, it can be incident to fewer than
$2^{i+1}$ peeled in-edges.  Across all active tenures of vertices in $L(x)$, at most
$2^{i+1}r(x)$ deletion events can contribute to $x$'s decrease counter.

Every opening of $x$ is triggered by the condition
\[
  \mathsf{Decrease}_i(x,t)\ge\lambda_i.
\]
Synchronization resets the counter for all currently active vertices below $x$.  Moreover,
\cref{lem:decrease-exact} ensures that the events charged to distinct openings are disjoint.
Hence the number of openings is at most
\[
  1+\frac{2^{i+1}r(x)}{\lambda_i}.
\]
If $\lambda_i=1$, then
$\lfloor 2^i/(4(H+1))\rfloor\le1$, which implies $2^i<8(H+1)$.  If
$\lambda_i\ge2$, then
$\lambda_i\ge 2^i/(8(H+1))$.  In either case
$2^{i+1}/\lambda_i=O(H+1)$, proving the claim.
\end{proof}

\begin{lemma}[Queries within one tree]\label{lem:one-tree-queries}
For a fixed bucket $i$, all decrease-counter evaluations and all exact leaf-degree evaluations use
$O(nH^2)$ cut queries over the entire execution.
\end{lemma}

\begin{proof}
The root counter is evaluated once after each peel, for $O(n)$ evaluations.

For $0\le j<H$, consider nodes whose subtrees contain $2^j$ padded leaves.  A nonroot node of
this height is queried only when its parent is opened.  Its parent has at most $2^{j+1}$ real
leaves, and by \cref{lem:opening-bound} is opened $O(H2^{j+1})$ times.  There are
$N/2^j$ nodes at height $j$, where $N=2^H=O(n)$. Thus all nodes at this height are
queried at most
\[
  \frac{N}{2^j}\cdot O(H2^{j+1})=O(NH)=O(nH)
\]
times.  Summing over $H$ heights gives $O(nH^2)$ decrease-counter evaluations.  Each evaluation
uses $O(1)$ cut queries by \cref{lem:source-prefix} and \eqref{eq:decrease-def}.

A real leaf has $r(x)=1$, so \cref{lem:opening-bound} implies that it is opened $O(H)$ times.
There are $n$ real leaves, and each exact degree computation uses $O(1)$ queries.  The resulting
$O(nH)$ additional queries are dominated by $O(nH^2)$.
\end{proof}

\begin{proof}[Proof of \cref{thm:topological-main}]
Correctness follows from \cref{lem:topological-correctness}.  We count queries.

The initial in-degrees require $n$ cut queries.  By \cref{lem:one-tree-queries}, each of the $H$
bucket trees contributes $O(nH^2)$ queries, for $O(nH^3)$ total.

It remains to account for offset initialization when a vertex enters a new bucket.  Each entry
computes one $O(1)$-query offset for each of $H+1$ ancestors.  A vertex moves only to lower bucket
indices and therefore enters at most $H$ buckets over the whole execution.  All offset
initializations consequently use $O(nH^2)$ queries.

Since $H=O(\log n)$, the total is $O(n\log^3 n)$.
\end{proof}

\begin{remark}[Local computation]\label{rem:local-time}
The proof optimizes oracle calls, not running time.  A direct implementation stores one history per
tree node and one offset for each active vertex--ancestor pair, using $O(n\log n)$ words.  Explicitly
resetting all offsets below an opened node still gives polynomial local running time.  More refined
lazy tags can reduce this overhead, but are unnecessary for the query bound.
\end{remark}

\section{Single-Source Reachability}\label{sec:reachability}

We now prove \cref{thm:reach-main}.  Throughout this section, $G$ is a DAG and
$\pi:V\to\{1,\ldots,n\}$ is a fixed topological ordering, so every arc $(u,v)$ satisfies
$\pi(u)<\pi(v)$.

\subsection{Counting and finding an outgoing edge}\label{subsec:find-edge}

For a vertex $v$ and a set $S\subseteq V\setminus\{v\}$, define the portion of $S$ after $v$ by
\[
  S_v^+ \defeq \{u\in S:\pi(v)<\pi(u)\}.
\]

\begin{lemma}[Outgoing edge count]\label{lem:outgoing-count}
For every $v\in V$ and $S\subseteq V\setminus\{v\}$,
\begin{equation}\label{eq:outgoing-count}
  |E(v,S)|
  = \cut(\{v\})+\cut(S_v^+)-\cut(S_v^+\cup\{v\}).
\end{equation}
Thus $|E(v,S)|$ can be computed with at most three cut queries.
\end{lemma}

\begin{proof}
No edge leaves $v$ for a vertex preceding $v$, so $E(v,S)=E(v,S_v^+)$.  No edge enters $v$ from
a vertex following $v$, so $E(S_v^+,v)=\varnothing$.  Applying
\cref{lem:symmetric-crossing} to $\{v\}$ and $S_v^+$ gives \eqref{eq:outgoing-count}.
\end{proof}

The minus sign in \eqref{eq:outgoing-count} is important: the union cut removes edges from the two
sets to their common exterior, leaving only edges crossing between them.

\begin{lemma}[Finding one outgoing edge]\label{lem:find-outneighbor}
Given $v\in V$ and $S\subseteq V\setminus\{v\}$, one can either certify that $E(v,S)=\varnothing$
or return a vertex $u\in S$ with $(v,u)\in E$ using $O(\log(|S|+1))$ cut queries.
\end{lemma}

\begin{proof}
First evaluate $|E(v,S)|$ using \cref{lem:outgoing-count}.  If it is zero, return that no edge
exists.  Otherwise split $S$ into two parts $S_1,S_2$ of sizes differing by at most one and evaluate
$|E(v,S_1)|$.  Recurse on $S_1$ if the value is positive, and otherwise recurse on $S_2$.  The
candidate set shrinks by a factor of at least two at every step, and a singleton positive set
identifies an out-neighbor.  Each step uses $O(1)$ queries.
\end{proof}

\subsection{Depth-first search through the oracle}\label{subsec:dfs}

The reachability algorithm is an ordinary depth-first search whose adjacency-list operation is
implemented by \cref{lem:find-outneighbor}.  It maintains a global visited set $R$.

\begin{algorithm}[H]
\caption{$\mathsf{Explore}(v)$}\label{alg:explore}
\begin{algorithmic}
  \State Insert $v$ into the global visited set $R$.
  \While{$|E(v,V\setminus R)|>0$}
    \State Use \cref{lem:find-outneighbor} to find $u\in V\setminus R$ with $(v,u)\in E$.
    \State $\mathsf{Explore}(u)$.
  \EndWhile
\end{algorithmic}
\end{algorithm}

The algorithm initializes $R=\varnothing$, calls $\mathsf{Explore}(s)$, and returns $R$.

\begin{lemma}[Reachability correctness]\label{lem:reach-correctness}
At termination, $R=\Reach_G(s)$.
\end{lemma}

\begin{proof}
Every inserted vertex is reached by following an arc from an already reached vertex, so
$R\subseteq\Reach_G(s)$.

For the reverse inclusion, suppose some reachable vertex $w$ is absent at termination.  Fix a path
$s=v_0,v_1,\ldots,v_k=w$ and let $v_j$ be the first vertex on this path absent from $R$.  Then
$v_{j-1}\in R$.  When the call $\mathsf{Explore}(v_{j-1})$ terminated, the vertex $v_j$ was still
unvisited, and the arc $(v_{j-1},v_j)$ showed that
$|E(v_{j-1},V\setminus R)|>0$.  This contradicts the termination condition of the while loop.
\end{proof}

\begin{proof}[Proof of \cref{thm:reach-main}]
Every recursive call inserts a previously unvisited vertex, so there are at most $n$ calls.  Every
successful while-loop iteration discovers a new vertex and therefore occurs at most $n-1$ times.
Each such iteration invokes \cref{lem:find-outneighbor} and uses $O(\log n)$ queries.  In addition,
there is one final zero test for every recursive call and one positive test for every successful
iteration, each costing $O(1)$ queries by \cref{lem:outgoing-count}.  The total is
$O(n\log n)$.
\end{proof}

\begin{proof}[Proof of \cref{cor:reach-from-scratch}]
First run \cref{alg:topological}.  On a DAG it returns a topological ordering using
$O(n\log^3 n)$ queries by \cref{thm:topological-main}.  Then run
\cref{alg:explore}, which uses $O(n\log n)$ additional queries by \cref{thm:reach-main}.
\end{proof}

\section*{Acknowledgments}

The authors developed all ideas, arguments, and proofs and prepared the manuscript in full detail. ChatGPT 5.6 was subsequently used solely for stylistic polishing; it did not contribute to the paper’s technical content.

\bibliographystyle{alpha}
\bibliography{references}

\end{document}